\documentclass{LMCS}

\def\dOi{10(1:14)2014}
\lmcsheading%
{\dOi}
{1--16}
{}
{}
{Jan.~19, 2013}
{Feb.~18, 2014}
{}


\ACMCCS{[{\bf Theory of computation}]:  Logic---Modal and temporal logics}

\usepackage{hyperref}
\usepackage{latexsym}
\usepackage{amsmath, amsthm, amssymb}
\usepackage{color}
\usepackage{graphicx,,graphics,subfig}

\DeclareMathSymbol{\Not}{\mathrel}{symbols}{"36}

\DeclareMathSymbol{\not}{\mathrel}{symbols}{"36}
\def\NEQ{\Not=} 

\def\nat{\mathbb{N}}
\def\real{\mathbb{R}}

 \def\vp{\varphi}
 \def\mM{\mathcal{M}}
\def\mL{\mathcal{L}}
\def\mE{\mathcal{E}}
\def\mC{\mathcal{C}}
\def\Oc{\mathit{O}}

\def\Inf{\mathit{INF}^{\neg\beta_1}}
\def\INF{\mathit{INF}}

\def\form{\mathit{Form}}
\def\TL{\mathit{TL}}
\def\TLs{\mathit{TL}(\until,\since,\suntil,\ssince)}
\newcommand     {\KMINUS}   {\textbf{\textsf {K}} ^{-}}
\newcommand     {\KPLUS}    {\textbf{\textsf {K}} ^{+}}
\newcommand{\until}{{\sf Until}}

\newcommand{\since}{{\sf Since}}
\newcommand{\suntil}{{\sf Until}^s}
\newcommand{\ssince}{{\sf Since}^s}
\newcommand{\G}{{\square}}

\newcommand{\Gb}{\overleftarrow{{\square}}}

\def\andl           {\wedge}                
\def\AND            {\bigwedge}
\def\orl            {\vee}                  

\def\not            {\neg}
\def\E              {\exists}
\def\A              {\forall}

\def\true              {\text{True}}
\def\false              {\text{False}}
\def\cond              {\mathit{Cond}}

\def\In1{\  \text{In}_{(d,c)}}
\newcommand{\EA}    { \overrightarrow{\exists}  \forall}

\newcommand{\DEA}    {\vee\overrightarrow{\exists} \forall}
\newcommand{\Aub}[2]    {(\A {#1})^{<{#2}}}                         

\newcommand{\Intprop}[2] {\Aub y {#2} {#1}(y)}              

\newcommand{\Alub}[3] {(\A {#1})_{> {#2}}^{< {#3}}}     

\newcommand{\intprop}[3] {\Alub y {#2} {#3} {#1}(y)}    

\newcommand{\Elub}[3] {(\E {#1})_{> {#2}}^{< {#3}}}     

\numberwithin{subcase}{case}

\newcommand{\emp}[1]    {\textbf {\textit{#1}}}

\def\FOMLO      {First-Order Monadic Logic of Order}

\def\FOMLOi     {\textit{FOMLO}}
\def\FOMLOiB    {\textit{FOMLO} }

\newcommand{\IFOMLO}{\mathit{FOMLO}}

\def\mI{\mathcal{I}}
\newcommand{\g}[1]      {\mathcal{#1}}

\begin{document}
\title[A Proof of Kamp's theorem]{A Proof of Kamp's theorem}

\author{Alexander Rabinovich}

\address{The Blavatnik  School \  of Computer Science, Tel Aviv University}
        \email{ rabinoa@post.tau.ac.il}


\begin{abstract}
      We provide   a  simple proof  of Kamp's theorem.
\end{abstract}
\maketitle

\section{Introduction}\label{sect:intro}


Temporal Logic ($\TL$) introduced  to Computer Science by Pnueli in
\cite{Pnu77}
 is a convenient framework for 
reasoning about ``reactive'' systems.  This has  made temporal logics  a
popular subject in the Computer Science community, enjoying
extensive research in the past 30 years.
In $\TL$ we describe basic system properties by {\it atomic
propositions\/} that hold at some points in time, but not at others.
More complex properties are expressed by formulas built from the
atoms using Boolean connectives and {\it Modalities\/} (temporal
connectives):
A $k$-place modality $M$ transforms statements $\varphi_1,\dots,
\varphi_k$ possibly on `past' or `future' points to a statement
$M(\varphi_1,\dots,\varphi_k)$ on the `present' point $t_0$.
The rule to determine the truth of a statement
$M(\varphi_1,\dots,\varphi_k)$ at $t_0$ is called a {\it truth
table\/} of $M$.
 The choice of particular modalities with their truth
tables yields different temporal logics. A temporal logic with
modalities   $M_1,\dots, M_k$ is denoted by
 $\TL(M_1,\dots ,M_k)$.

The simplest example is the one place modality $\lozenge P$ saying:
``$P$  holds some time in the future.'' Its truth table  is  
formalized by $\varphi_{_{\lozenge}} (x_0,P):=  \exists x(x>x_0\wedge
P(x))$. This is a formula  of the First-Order  Monadic Logic of
Order (\FOMLOi) - 
a fundamental formalism in Mathematical Logic
  where formulas are built using atomic propositions $P(x)$,
   atomic relations between elements $x_1=x_2$, $x_1<x_2$, Boolean
connectives and  first-order quantifiers $\exists x$  and $\forall
x$.
Two
more natural modalities are the modalities $\until $
(``\emph{Until}'') and $\since $ (``\emph{Since}''). $X\until Y$ means
that $X$ will hold from now until a time in the future when $Y$
will hold. $X\since Y$ means that $Y$ was true at some point of time
in the past and since that point $X$ was true until (not
necessarily including) now.
 Both
modalities have truth tables in \FOMLOi.
 Most modalities used in the literature are defined by such
\FOMLOi \
truth tables, and as a result, every temporal formula 
translates directly into an equivalent  $\IFOMLO$ formula. Thus, the
different temporal logics may be considered as a convenient way to use
fragments of $\IFOMLO$.  $\IFOMLO$  can also
serve as a yardstick by which  one is able to check the strength of temporal logics: 
A temporal logic is {\it expressively complete\/} for a fragment $L$
of $\IFOMLO$ if every formula of $L$ with a single free variable
$x_0$ is equivalent to a temporal formula.

Actually, the notion of expressive completeness   refers to a temporal logic and to a model (or a class of models),
 since the
question whether two formulas are equivalent depends on the domain
over which they are evaluated.  Any  (partially) ordered set with
monadic predicates is a model for $\TL$ and \FOMLOi, but the main,
{\it canonical\/}, linear time intended models  are the non-negative
integers $\langle \nat,~<\rangle$  for discrete time and the
 reals $\langle \real,<\rangle$  for continuous time.

 Kamp's theorem \cite{Kam68}
  states that the temporal logic with modalities $\until$ and $\since$
is expressively
complete for $\IFOMLO$ over the above two linear time  canonical\footnote{the technical notion which unifies  $\langle \nat,~<\rangle$  and
$\langle \real,<\rangle$ is Dedekind completeness.}
models.
\begin{quote}
This  seminal theorem  initiated the whole study of expressive
completeness,  and it remains one of the most interesting and
distinctive results in temporal logic; very few, if any, similar
`modal' results exist. Several alternative proofs of it and stronger
results have appeared; none of them are trivial (at least to most
people)  \cite{HR06}.
\end{quote}
The objective of this paper is to provide a simple proof of Kamp's  theorem.

The rest of the paper is organized as follows: In Section
\ref{sect:prel} we recall the definitions of the monadic logic,
the temporal logics and state  Kamp's  theorem.
Section  \ref{sect:ea-formulas}   introduces formulas in a normal form and states their simple properties. In Section  \ref{sect:kamp} we prove  Kamp's theorem.
The proof of one proposition is postponed to Section \ref{sect:neg}.
Section \ref{sect:related}  comments  on  the   previous  proofs of Kamp's theorem.
Finally, in Section \ref{sect:future}, we  show that our proof can be easily modified to prove expressive completeness for the future fragment
of $\IFOMLO$.

\section{Preliminaries}\label{sect:prel}
In this section
 we recall the definitions of    the first-order  monadic logic of order,
the temporal logics and state  Kamp's  theorem.

 Fix a set $\Sigma$ of \emph{atoms}. We use
$P, Q, R, S \dots$ to denote members of $\Sigma$.
The syntax and semantics of both logics are defined below with respect
to such  $\Sigma$.


\subsection{\FOMLO} \label{subsect:fomlo}


\emp{Syntax:} In the context of \FOMLOi, the atoms of $\Sigma$ are
referred to (and used) as \emph{unary predicate symbols}. Formulas
are built using these symbols,  plus two binary relation symbols:
$<$
and $=$, and a  set of {first-order  variables} 
(denoted: $ x,y, z, \dots $).
Formulas are defined by the grammar:
$$atomic ::= ~~x < y ~~|~~ x = y ~~|~~P(x)~~~~~~~~~~ (\mbox{where} ~~ P \in \Sigma)$$
$$\varphi ::= ~~ atomic ~~|~~ \neg \varphi_1 ~~|~~ \varphi_1 \orl \varphi_2 ~~|~~ \varphi_1 \andl \varphi_2 ~~|~~ \exists x \varphi_1 ~~|~~ \forall x \varphi_1$$
We  also use the standard abbreviated notation for
\emp{bounded quantifiers}, e.g., $(\exists x)_{> z}(\dots)$ denotes
$\exists x ((x > z) \andl (\dots))$, and $(\forall x)^{<
z}(\dots)$ denotes $\forall x ((x < z) \rightarrow (\dots))$,  and  $(\Alub{x}{z_1}{z_2}
 (\dots)$ denotes $\forall x ((z_1<x < z_2) \rightarrow (\dots))$,
etc.

\noindent \textbf{{Semantics.}}
Formulas are interpreted over \textit{labeled  linear orders} which are  called \emph{chains}. A $\Sigma$-\emp{chain}
 is a triplet $\mM = ( {T}, <, \mI)$ where
${T}$ is a set - the \emph{domain} of the chain, 
$<$ is a linear  order relation on ${T}$, 
and $\g{I} : {\Sigma} \rightarrow \g{P} ({T})$ is the
\emph{interpretation} of  $\Sigma$ (where $\g{P}$ is the powerset
notation). We use the standard notation $\mM, t_1, t_2, \dots t_n
\models \varphi(x_1, x_2, \dots x_n)$ to indicate that the formula $\varphi$ with free variables among $x_1,\dots,x_n$
is satisfiable in $\mM$ when $x_i$ are interpreted as elements $t_i$ of $\mM$.
For atomic $P(x)$ this is defined by: $\mM, t \models P(x)$ iff $t
\in \g{I}(P)$; the semantics of $<, =, \neg, \andl, \orl, \exists
\mbox { and } \forall$ is defined in a standard way.


\subsection{$\TL(\until,\since)$ Temporal Logic } \label{subsect:TL}

In this section we recall  the syntax and semantics of a temporal logic with
\emp{strict}-$\until$ and \emp{strict}-$\since$   modalities, denoted by $\TL(\until,\since)$.

In the context of temporal logics, the atoms of $\Sigma$ are
used as \emph{atomic propositions} (also called \emph{propositional
atoms}). Formulas of  $\TL(\until,\since)$ are built using these atoms, Boolean connectives  and
\emp{strict}-$\until$ and \emp{strict}-$\since$ binary  modalities.
The formulas are
  defined by the grammar:
$$F ::= ~~ \true~~|~~P ~~|~~ \neg F_1 ~~|~~ F_1 \orl F_2 ~~|~~ F_1 \andl F_2 ~~|~~ F_1\until F_2~~| ~~F_1\since F_2,$$
where $P \in \Sigma$. 

\noindent \textbf{{Semantics.}} 
Formulas are interpreted at \emph{time-points} (or \emph{moments}) in
chains  (elements of the domain).
The semantics of  $\TL(\until,\since)$ formulas is   defined inductively:
Given a chain  $\mM = ( {T}, <, \mI)$ and
 $t \in
T$,
 define when a 
 formula $F$ \emph{holds} in $\mM$ at $t$ - denoted $\mM, t \models F$:
\begin {itemize}
    \item   $\mM, t \models P$ iff $t \in \g{I}(P)$, for any propositional atom $P$.
    \item   $\mM, t \models F_1 \orl F_2$ iff $\mM, t \models F_1$ or $\mM, t \models F_2$;
                similarly  for $\andl$ and $\neg$.
\item $\mM, t \models F_1\until  F_2 $    iff there is $t'> t$ such that
   $\mM,t'\models F_2$ and $\mM, t_1 \models F_1$ for all $t_1\in (t,t')$.

\item $\mM, t \models F_1\since   F_2 $    iff there is $t'< t$ such that
   $\mM,t'\models F_2$ and $\mM, t_1 \models F_1$ for all $t_1\in (t',t)$.
\end {itemize}
We will use standard abbreviations. As usual
$\G  F$ (respectively, $\Gb F  $)  is an abbreviation for $\neg( \true\until (\neg F))$ (respectively, $\neg( \true\since (\neg F))$), and
 $\KPLUS(F)$   (respectively,
  $\KMINUS(F)$)   is an abbreviation for  $\neg((\neg F) \until \text{True})$
(respectively,  $\neg((\neg F) \since \text{True})$).
\begin{enumerate}
\item $\G F $  (respectively, $\Gb F $) holds at $t$  iff  $F$ holds everywhere after (respectively, before)  $t$.
 \item  $\KMINUS(F)$   holds at a moment $t$ iff
$t=\sup(\{t' \mid t'<t\mbox{ and } F \mbox{ holds at } t'\})$.

 \item  
   $\KPLUS(F)$   holds at a moment $t$ iff
$t=\inf(\{t' \mid t'>t\mbox{ and } F \mbox{ holds at } t'\})$.
\end{enumerate}
Note that $\KPLUS (\true)$ (respectively, $\KMINUS(\true)$) holds at
 $t$ in $\mM$ if $t$ is a right limit (respectively, a left limit) point of the underlining
order.
In particular, both  $\KPLUS (\true)$ and $\KMINUS(\true)$ are equivalent to $\false$ in  the chains     over $(\nat,<)$,


\subsection{Kamp's   Theorem} \label{subsect:exComp}


{Equivalence} 
between temporal and monadic formulas is naturally defined: $F$ is
equivalent to $\varphi(x)$  over a class $\mC$ of structures iff for
any $\mM\in \mC$   and $t \in \mM$: $\mM, t \models F
\Leftrightarrow \mM, t \models \varphi(x)$. If $\mC$ is the  class of
all chains, we will say that $F$ is equivalent to $\varphi$.

A linear order $(T,<)$ is \emph{Dedekind complete}
if  for every non-empty subset $S$ of $T$, if $S$ has a lower bound
in $T$ then it has a greatest lower bound, written $\inf(S)$, and if
$S$ has an upper bound in $T$ then it has a least upper bound,
written $\sup(S)$.
The canonical linear time models $(\nat,<)$ and $(\real,<)$   are Dedekind complete, while
the order of the  rationals is not Dedekind complete. A chain is Dedekind complete if its underlying linear order
is Dedekind complete.

The fundamental theorem  of Kamp's
states that
$\TL(\until,\since)$ is expressively equivalent to \FOMLOiB over {Dedekind complete } chains.  
\begin{thm} [Kamp \cite{Kam68}] \label{th:kamp}
     \begin{enumerate}
     \item Given any $\TL(\until,\since)$  formula $A$ there is a  $\FOMLOiB$ formula $\varphi_A(x)$ which is equivalent to $A$
     over all  chains.
     \item Given any $\FOMLOiB$ formula $\varphi(x)$ with one free variable,
     there is a $ \TL(\until,\since)$  formula   which is equivalent to $\varphi$
     over Dedekind complete chains.
     \end{enumerate}
\end{thm}

\noindent The meaning preserving translation from $\TL(\until,\since)$  to  $\FOMLOiB$
is easily obtained by structural induction.
The   contribution of our paper is a proof of Theorem \ref{th:kamp}
(2).
 The proof is constructive. An algorithm which for
every $\IFOMLO$ formula $\varphi(x )$ constructs a
$\TL(\until,\since)$ formula which is equivalent to $\varphi$ over
Dedekind complete chains is  easily extracted from our proof.
However, this  algorithms is not efficient in the sense of
complexity theory.
This is unavoidable because there is a non-elementary succinctness
gap between $\IFOMLO$ and $\TL (\until , \since)$ even over the
class of finite chains, i.e.,  for every $ m,n\in \nat$ there is a
$\IFOMLO$ formula $\varphi(x_0)$ of
 size $|\varphi|>n$ which is not  equivalent (even over finite
chains) to any $\TL(\until,\since) $ formula of size $\leq \exp(m,
|\varphi|)$, where the $m$-iterated  exponential function $\exp(m,n)
$ is defined  by induction on $m$ so that $\exp(1,n)=2^n$, and
$\exp(m+1,n)=2^{\exp(m,n)}$.

\section{ ${\vec{\exists}} \forall $   formulas} \label{sect:ea-formulas}
First, we introduce $\EA$ formulas which are instances of the
Decomposition formulas of  \cite{GPSS80}.
\begin{defi} [$\EA$-formulas] %
    \label{def:simple}
    Let $\Sigma$ be a set of monadic predicate names.
    An  {$\EA$-formula} over $\Sigma $  is a formula of the form:
\begin{align*} \displaystyle
\psi(z_0,\dots, z_m) & := 
 \E x_n \dots \E x_1 \E x_0                                                                                                  \\
                        & \left(\bigwedge_{k=0}^m z_k=x_{i_k}
                        \right) \wedge
                          \left(x_n > x_{n-1} > \dots > x_1 > x_0 \right) ~~~~~~\mbox{``ordering of $x_i$ and $z_j$}"                                                 \\
                        & \andl~\AND _{j=0}^{n} \alpha_j(x_j)
                        \hspace{3.cm}~~~~~~~~~~~~~~~~~~~~~~~~~``\mbox{Each $\alpha_j$ holds at}~x_j"                             \\  
                        & \andl~\AND _{j=1}^{n} [\intprop {\beta_j} {x_{j-1}} {x_{j}}]                                    
                          \hspace{2.5cm}~~~~``\mbox{Each $\beta_j$ holds along}~(x_{j-1}, x_j)"                                                    \\
                         & \andl~(\forall y)_{>x_n} {\beta_{n+1}} (y)                                                                                              
                       \hspace{2.cm}~~~~~~~~~~~~``\beta_{n+1}~\mbox{holds everywhere  after}~x_n"                                      \\
                        &  \andl ~(\forall y)^{<x_0}
                        {\beta_0} (y)                                                                                              
                        \hspace{2.cm}~~~~~~~~~~~~~~~~``\beta_0~\mbox{holds everywhere before} ~x_0"                                      
\end{align*}
with a prefix of $n+1$ existential quantifiers
 and
with all $\alpha_j$, $\beta_j$  quantifier free formulas with one variable  over
$\Sigma$, and $i_0,\dots,i_m\in \{0,\dots,n\}$. ($\psi$   has $m+1$ free variables $z_0,\dots, z_{m}$ and
$n+1$ existential quantifiers, $m+1$ quantifiers    are dummy and are
introduced just in order to simplify notations.)
\end{defi}
It is clear that
\begin{lem}\label{lem:triv}\hfill
\begin{enumerate}

\item  Conjunction of   $\EA$-formulas  is equivalent to a disjunction of $\EA$-formulas. 
\item Every  $\EA$-formula  
is equivalent to a conjunction of $\EA$-formulas    with at most two
free variables.
\item For every $\EA$-formula $\vp$ 
 the formula  $\exists x \vp$ is equivalent to a
$\EA$-formula. 
\end{enumerate}
\end{lem}
\begin{defi} [$\DEA$-formulas] A formula is a $\DEA$ formula if it is equivalent to a disjunction of
$\EA$-formulas.
\end{defi}
 \begin{lem}[closure properties]\label{lem:closure}
The set of $\DEA$ formulas is  closed under disjunction,
conjunction, and existential quantification.
\end{lem}
\begin{proof} By (1) and (3) of Lemma \ref{lem:triv}, and distributivity of $\exists$ over $\vee$.
\end{proof}
The set of $\DEA$  formulas is not closed under
negation\footnote{The truth table of $P\until Q$ is  an $\EA$ formula 
$(\exists x') _{> x} (Q(x') \andl (\forall y) _{> x} ^{< x'} P(y))$,
yet we can prove that its negation is not equivalent to any $\DEA$ formula.}. However,
we show   later (see Proposition \ref{prop:fo2ea}) that the negation
of a  $\DEA$  formula is equivalent to a $\DEA$  formula in the expansion
of the  chains by all
 $\TL(\until,\since)$ definable predicates.

The  $\DEA$  formulas with one free variable can be easily
translated to
 $\TL(\until,\since)$.

\begin{prop}[From $\DEA$-formulas to $\TL(\until,\since)$ formulas]
\label{prop:form} 
Every $\DEA$-formula     with one free  variable  is equivalent
  to a  $\TL(\until,\since)$ formula.
\end{prop}
\proof
By a simple formalization we show that every   $\EA$-formula with
one free  variable  is equivalent
 to a  $\TL(\until,\since)$ formula. This immediately implies the
 proposition.

 Let $\psi(z_0)$ be an  $\EA$-formula
$$         \E x_n \dots \E x_1 \E x_0
~  z_0=x_k                          \andl~ \left(x_n > x_{n-1} > \dots
> x_1 > x_0\right)
                         \wedge ~\AND _{j=0}^{n} \alpha_j(x_j)$$

                          $$\wedge~
                             \AND _{j=1}^{n}  \intprop {\beta_j}
                             {x_{j-1}}  {x_{j}}
                          \wedge~\Intprop {\beta_0} {x_{0}}                                                                                                             %
                           \wedge ~(\forall y)_{>x_n}
                        {\beta_{n+1}} (y)
$$
Let $A_i$ and $B_i$ be temporal formulas equivalent to $\alpha_i$
and $\beta_i$ ($A_i$ and $B_i$ do not even use $\until$ and $\since$
modalities).
   It is
easy to see that $\psi$ is equivalent to the conjunction of
\[ A_k\wedge (B_{k+1}\until
(A_{k+1} \wedge (B_{k+2}\until  \cdots( A_{n-1}\wedge (B_n \until
(A_n \wedge \G B_{n+1}))\cdots ))))\]
 and
\[ A_k\wedge (B_{k-1}\since
(A_{k-1} \wedge (B_{k-2}\since (  \cdots A_{1}\wedge (B_{1} \since
(A_{0} \wedge \Gb B_0))\cdots ))  \]
\section{Proof of Kamp's theorem}\label{sect:kamp}
The next definition plays a major role in the proof   Kamp's   theorem  \cite{GPSS80}.

\begin{defi} Let $\mM$ be a $\Sigma$ chain. We denote by  $\mE[\Sigma]$  the set of unary predicate names
$\Sigma\cup\{A~\mid A$ is an $\TL(\until,\since)$-formula over $\Sigma ~\}$.
 The canonical $\TL(\until,\since)$-expansion of $\mM$ is  an expansion of
$\mM$ to an $\mE[\Sigma]$-chain,  where each predicate name $A\in
\mE[\Sigma]$ is interpreted as $\{a\in \mM\mid \mM,a\models A\}$\footnote{ We often  use  ``$a\in \mM$''
instead of ``$a$ is an element of the domain of $\mM$''}. We say
that first-order formulas in the signature $\mE[\Sigma] \cup\{<\}$
are equivalent over $\mM$ (respectively, over a class of
$\Sigma$-chains $\mC$)  if they are equivalent in the canonical
expansion of $\mM$ (in  the canonical expansion of every $\mM\in \mC$).
\end{defi}

Note that if $A$ is a $\TL(\until,\since)$ formula  over $\mE[\Sigma]$ predicates,
then it is equivalent to a $\TL(\until,\since)$  formula over $\Sigma$, and hence
to an atomic formula in the canonical $\TL(\until,\since)$-expansions.

In this section  and the next one we say that ``formulas are equivalent  in a
 chain $\mM$''   instead of ``formulas are equivalent  in the
canonical $\TL(\until,\since)$-expansion of $\mM$.'' The $\EA$ and $\DEA$ formulas
are defined as previously, but now they can use as atoms $\TL(\until,\since)$
definable predicates.

It is clear that all the results stated above 
hold for this modified notion of $\DEA$ formulas. In particular,
every $\DEA$ formula with one free variable is equivalent to an
$\TL(\until,\since)$ formula, and the set of $\DEA$ formulas is closed under
conjunction, disjunction and existential quantification. However, now
the set of $\DEA$ formulas is also closed under negation, due to
 the next proposition whose proof is postponed to Sect.
\ref{sect:neg}.

\begin{prop}(Closure under negation) \label{lem:neg}
 The negation of  $\EA$-formulas 
   with at most two free  variables
is equivalent over Dedekind complete chains  to a disjunction of
 $\EA$-formulas.
\end{prop}
As a consequence we obtain
\begin{prop}\label{prop:fo2ea}
 Every first-order formula
is equivalent over Dedekind complete chains  to a disjunction of
$\EA$-formulas.
\end{prop}
\begin{proof} We proceed by   structural induction.
\begin{description}
\item[Atomic]
It is clear that every atomic formula is equivalent to a disjunction
of  (even quantifier free)  $\EA$-formulas.

\item[Disjunction] - immediate.

\item[Negation]
If $\vp$ is an $\EA$-formula, then by Lemma \ref{lem:triv}(2) it is
equivalent to a conjunction of $\EA$ formulas with at most two free
variables. Hence, $\neg\vp$ is equivalent to a  disjunction of $\neg
\psi_i$ where $\psi_i$ are $\EA$-formulas with at most two free
variables. By
 Proposition
\ref{lem:neg}, $\neg \psi_i$ is equivalent to a disjunction of $\EA$
formulas $\gamma_i^j$. Hence, $\neg\vp$ is equivalent to a
disjunction $\vee_i\vee_j \gamma_i^j$ of $\EA$ formulas.

If $\vp$ is a disjunction of $\EA$ formulas $\vp_i$, then $\neg\vp$
is equivalent to the conjunction of $\neg\vp_i$. By the above,
$\neg\vp_i$ is equivalent to a  $\DEA$ formula. Since, $\DEA$
formulas are closed under conjunction (Lemma \ref{lem:closure}), we
obtain that $\neg\vp$ is equivalent to a disjunction of $\EA$
formulas.

\item[$\exists$-quantifier]
For $\exists$-quantifier, the claim follows from Lemma \ref{lem:closure}.  \qedhere
\end{description}
\end{proof}

Now, we are ready to prove Kamp's Theorem:
\begin{thm}
For every  $\FOMLOiB$ formula $\varphi(x)$ with one free variable,
     a $ \TL(\until,\since)$  formula exists  that is equivalent to $\varphi$
     over Dedekind complete chains.
\end{thm}
\begin{proof}
By Proposition \ref{prop:fo2ea}, $\vp(x)$ is equivalent over
Dedekind complete chains to a disjunction of $\EA$ formulas
$\vp_i(x)$. By Proposition  \ref{prop:form}, $\vp_i(x)$ is
equivalent to a $\TL(\until,\since)$ formula. Hence, $\vp(x)$ is
equivalent over Dedekind complete chains  to a $\TL(\until,\since)$
formula.
\end{proof}
This completes our proof of Kamp's theorem except Proposition \ref{lem:neg} which is  proved in the next section.
\section{Proof of Proposition  \ref{lem:neg}}\label{sect:neg}
 Let $\psi(z_0,z_1)$ be an  $\EA$-formula
\begin{gather*}        \E x_n \dots \E x_1 \E x_0
 [ z_0=x_m \wedge z_1=x_k                         \andl~ \left(x_0 < x_{1} < \dots
< x_{n-1} < x_n \right)
                         \wedge ~\AND _{j=0}^{n} \alpha_j(x_j)\\
                           \wedge~
                             \AND _{j=1}^{n}  \intprop {\beta_j}
                             {x_{j-1}}  {x_{j}}
                          \wedge~\Intprop {\beta_0} {x_{0}}                                                                                                             %
                           \wedge ~(\forall y)_{>x_n}
                        {\beta_{n+1}} (y)]
\end{gather*}
We consider two cases. In the first case  $k=m$, i.e., $z_0=z_1$ and in the second $k\NEQ m$.

If $k=m$,  then $\psi$ is equivalent to $z_0=z_1\wedge \psi'(z_0)$,
where $\psi'$ is an  $\EA$-formula. By Proposition  \ref{prop:form},
$\psi'$ is equivalent to an $\TL(\until, \since)$ formula $A'$.
Therefore, $\psi $ is equivalent to an  $\EA$-formula $\exists
x_0[z_0=x_0\wedge z_1=x_0\wedge A'(x_0)]$, and $\neg\psi$ is equivalent to a $\DEA$ formula $z_0<z_1\vee z_1<z_0 \vee  \exists
x_0[z_0=x_0\wedge z_1=x_0\wedge \neg A'(x_0)]$.

If $k\NEQ m$, w.l.o.g. we assume that $m<k$. Hence, $\psi$ is equivalent to a
conjunction of

 \begin{enumerate}
\item $\psi_0(z_0)$ defined as:\\
$   \E x_0 \dots \E x_{m-1} \E x_{m} [ z_0=x_m \andl~ \left(x_0 <
x_{1} < \dots < x_m  \right)
                         \wedge ~\AND _{j=0}^{m} \alpha_j(x_j)$
                          $$\wedge~
                             \AND _{j=1}^{m}  \intprop {\beta_j}
                             {x_{j-1}}  {x_{j}}
                          \wedge~\Intprop {\beta_0} {x_{0}}  ]                                                                                                           %
$$
\item $\psi_1(z_1)$ defined as:\\
$   \E x_k \dots \E x_{k+1} \E x_{n} [ z_1=x_k \andl~ \left(x_k <
x_{k+1} < \dots < x_n \right)
                         \wedge ~\AND _{j=k}^{n} \alpha_j(x_j)$
                          $$\wedge~
                             \AND _{j=k+1}^{n}  \intprop {\beta_j}
                             {x_{j-1}}  {x_{j}}
                           \wedge ~(\forall y)_{>x_n}
                         {\beta_{n+1}} (y)]
$$
\item $\vp(z_0,z_1)$ defined as:\\
 $         \E x_m \dots  \E x_k
 [                            \left(z_0=x_m < x_{m+1} < \dots
  < x_k=z_1 \right)
                         \wedge ~\AND _{j=m}^{k} \alpha_j(x_j)$
                                                   $$\wedge~
                             \AND _{j=m+1}^{k}  \intprop {\beta_j}
                             {x_{j-1}}  {x_{j}}]
$$
\end{enumerate}
The  first two formulas are $\EA$-formulas with one free variable.
Therefore, (by  Proposition  \ref{prop:form})  they are equivalent
to
  $\TL(\until,
\since)$ formulas (in the signature $\mE[\Sigma]$). Hence, their
negations are  equivalent (over the canonical expansions) to atomic (and hence to
$\EA$) formulas.

Therefore, it is sufficient to show that the negation of the third
formula is equivalent over  Dedekind complete chains to a
disjunction of $\EA$-formulas.
This is stated  in the following
   lemma:
\begin{lem}\label{lem:neg1}
The negation of any formula of the form
\begin{equation} \label{eq:1}        \E x_0 \dots  \E x_n
 [                            \left(z_0=x_0 <  \dots
  < x_n=z_1 \right)
                         \wedge ~\AND _{j=0}^{n} \alpha_j(x_j)
                                                    \wedge~
                             \AND _{j=1}^{n}  \intprop {\beta_j}
                             {x_{j-1}}  {x_{j}} ]
\end{equation}
where $\alpha_i,\beta_i$ are quantifier free, is equivalent (over
Dedekind complete chains) to a disjunction of  $\EA$-formulas.
\end{lem}
In the rest of this section we prove  Lemma
\ref{lem:neg1}. Our proof is organized as follows.  In Lemma \ref{lem:occ} we prove an instance  of  Lemma
\ref{lem:neg1} where $\alpha_0$, $\alpha_n$ and  all $\beta_i$ are equivalent to $\true$. Then we derive a more general  instance (Corollary \ref{cor:cont})  where $\beta_n$ is equivalent to true.
Finally we prove the full version of   Lemma
\ref{lem:neg1}.

First, we introduce some helpful  notations.
\begin {nota} \label{denote:simple}
    We use the abbreviated notation
    $ [ \alpha_0, \beta_1 , \dots,   \alpha_{n-1}, \beta_n,\alpha_n](z_0,z_1)$
    for the $\EA$-formula
      as  in (\ref{eq:1}).
\end {nota}
In this notation Lemma \ref{lem:neg1} can be rephrased as $\neg [
\alpha_0, \beta_1,  \dots,   \alpha_{n-1}, \beta_n,\alpha_n](z_0,z_1)$
is equivalent (over Dedekind complete chains) to a $\DEA$ formula.

 We start with the instance of  Lemma \ref{lem:neg1} where all $\beta_i$ are $\true$.
\begin{lem} \label{lem:occ}
$\neg \E x_1 \dots  \E x_n
                             \left(z_0<x_1 <  \dots
  < x_n<z_1 \right)     \wedge ~\bigwedge_{i=1}^{n}  P_i(x_i)$ is
 equivalent over Dede\-kind complete chains  to a $\DEA$ formula  $\Oc_n  (P_1, \dots, P_{n},z_0,z_1)$.
\end{lem}
\begin{proof}
 We proceed by induction on $n$.

\emph{Basis}: $ \neg \Elub {x_1}{z_0}{z_1} P_1 (x_1) $ is equivalent
to $\Alub{y} {z_0}{z_1} \neg P_1(y)  $.

\emph{Inductive step}: $n\mapsto n+1$. We assume that a $\DEA$ formula  $O_n$
has already  defined and construct a $\DEA$ formula $O_{n+1}$.

Observe that if the interval  $(z_0,z_1)$ is non-empty, then   one of the following cases holds:
\begin{description}
\item[Case 1]
 $P_1$ does not occur in $(z_0,z_1)$, i.e.  $\Alub{y} {z_0}{z_1}\neg P_1( y)$.
 Then $\Oc_{n+1}(P_1, \dots, P_{n+1},z_0,z_1)$ should be equivalent to $\true$.
\item[Case 2] If case 1 does not hold  then let 
$r_0=\inf\{z\in (z_0,z_1)\mid P_1(z)\}$ (such $r_0$ exists by Dedekind completeness.
Note that $r_0=z_0$ iff $\KPLUS(P_1)(z_0)$. 
 If $r_0>z_0$ then $r_0\in (z_0,z_1)$ and
 $r_0$ is definable by the following  $\DEA$  formula:
\begin{align}
\label{eq:un0}
\INF(z_0,r_0,z_1,P_1):= &z_0<r_0<z_1 \wedge  \Alub {y}{z_0} {r_0}
\neg P_1(y) \wedge \notag \\ &
 ~~~~\wedge( P_1(r_0) \vee \KPLUS(
P_1)(r_0))
\end{align}
\begin{description}
\item[Subcase $r_0=z_0$] In this subcase $\Oc_n(P_2,\dots ,P_n,z_0,z_1)$  and  $\Oc_{n+1}(P_1, \dots, P_{n+1},z_0,z_1)$ should be equivalent.
\item[Subcase $r_0\in (z_0,z_1)$]
Now $\Oc_n(P_2,\dots ,P_n,r_0,z_1)$ and
$\Oc_{n+1}(P_1, \dots, P_{n+1},z_0,z_1)$ should be equivalent.
\end{description}
\end{description}
Hence,  $\Oc_{n+1}(P_1, \dots, P_{n+1},z_0,z_1)$  can be defined as
the  disjunction of ``$(z_0,z_1)$ is empty'' and  the following formulas:
\begin{enumerate}
\item
 $\Alub{y} {z_0}{z_1}\neg P_1((y)$
 \item
$\KPLUS(P_1)(z_0) \wedge \Oc_n(P_2,\dots ,P_n,z_0,z_1)$
\item
$\Elub{r_0} {z_0}{z_1}\big(\INF(z_0,r_0,z_1,P_1) \wedge
\Oc_n(P_2,\dots ,P_n,r_0,z_1)\big)$
\end{enumerate}
The first formula is
 a  $\DEA$ formula.
By the inductive assumptions $ \Oc_{n}$ is
 a  $\DEA$ formula.
 $\KPLUS(P_1)(z_0)$ is an atomic (and hence a $\DEA$) formula in the canonical expansion, and $ \INF(z_0,r_0,z_1,P_1)$ is
a $\DEA$ formula. Since  $\DEA$ formulas are closed under
conjunction, disjunction and the existential quantification, we conclude
that $\Oc_{n+1}$ is a $\DEA$ formula.
\end{proof}
As a consequence we obtain\newpage
\begin{cor} \label{cor:cont}\hfill
\begin{enumerate}
\item
$\neg\Elub{z}{z_0}{z_1} [\alpha_0, \beta_1, \alpha_1, \beta_2,
\dots,   \alpha_{n-1}, \beta_n,\alpha_n](z_0,z)$ over Dedekind complete chains is
 equivalent  to a $\DEA$ formula.

\item $\neg\Elub{z}{z_0}{z_1} [\alpha_0, \beta_1, \alpha_1, \beta_2,  \dots,   \alpha_{n-1}, \beta_n,\alpha_n](z,z_1)$ over Dedekind complete chains is
 equivalent  to a $\DEA$ formula.
\end{enumerate}
\end{cor}
\begin{proof} (1)
Define
\begin{align*}
F_n:=\alpha_n& \\
F_{i-1}:= \alpha_{i-1}\wedge (\beta_i \until F_i) &\qquad \text{ for
}i=1,\dots,n
\end{align*}
Observe that there is  $z\in (z_0,z_1)$   such that $ [\alpha_0,
\beta_1, \alpha_1, \beta_2,  \dots,   \alpha_{n-1},
\beta_n,\alpha_n](z_0,z)$  iff $F_0(z_0)$ and there is an increasing
sequence $x_1<\dots <x_n$ in an open interval $(z_0,z_1)$ such that
$F_i(x_i)$ for $i=1,\dots ,n$. Indeed, the direction $\Rightarrow$
is trivial. The  direction   $\Leftarrow$ is easily proved by
induction.

The \emph{basis} is trivial.

\emph{Inductive step:} $n\mapsto n+1$. Assume   $F_0(z_0)$ holds and that
$(z_0,z_1)$ contains an  increasing sequence $x_1<\dots <x_{n+1}$
such that $F_i(x_i)$ for $i=1,\dots ,n+1$. By the inductive
assumption  there is $y_1\in (z_0,x_{n+1})$ such that
  $$[\alpha_0, \beta_1, \alpha_1, \beta_2,  \dots,  \beta_{n-1} \alpha_{n-1}, \beta_n ,(\alpha_n\wedge \beta_{n+1}\until \alpha_{n+1})](z_0,y_1).$$

In particular,  $y_1$ satisfies $(\alpha_n\wedge \beta_{n+1}\until
\alpha_{n+1})$. Hence, there is $y_2>y_1$ such that $y_2$ satisfies
$\alpha_{n+1}$ and $\beta_{n+1}$ holds along $(y_1,y_2)$.

If $y_2\leq x_{n+1}$ then the  required $z\in (z_0,z_1)$ equals to $y_2$, and we
are done. Otherwise, $x_{n+1}< y_2$. Therefore, $x_{n+1}\in
(y_1,y_2)$ and $\beta_{n+1}$ holds along $(y_1,x_{n+1})$. Hence, the  required $z$ equals to $x_{n+1}$.

The above observation and Lemma \ref{lem:occ} imply that 
 $\neg F_0(z_0)\vee \Oc_n(F_1,\dots, F_n,z_0,z_1)$   is a  $\DEA$ formula that
is equivalent to
$\neg\Elub{z}{z_0}{z_1} [\alpha_0, \beta_1, \alpha_1, \beta_2,
\dots,   \alpha_{n-1}, \beta_n,\alpha_n](z_0,z)$.

(2) is the mirror image of (1) and is    proved similarly.
\end{proof}
Now we are ready to prove Lemma \ref{lem:neg1}, i.e.,
\begin{center} $\neg      [ \alpha_0, \beta_1  \dots,  \beta_{n-1}, \alpha_{n-1}, \beta_n ,\alpha_n](z_0,z_1)$
is equivalent\\ over Dedekind complete chains  to a $\DEA$ formula.
\end{center}
\begin{proof} (of Lemma \ref{lem:neg1})
If  the interval  $(z_0,z_1)$ is empty then the assertion is
immediate. We assume that $(z_0,z_1)$ is non-empty. Hence,  at least
one of the following cases holds:
\begin{description}
\item[Case 1] $\neg \alpha_0(z_0)$ or  $\KPLUS(\neg\beta_1)(z_0)$. 
\item[Case 2] $\alpha_{0}(z_0)$,  and $ \beta_{1}$  holds along
$(z_0,z_1)$.

\item[Case 3]
\begin{enumerate}
\item $ \alpha_0(z_0)\wedge  \neg\KPLUS(\neg\beta_1)(z_0)$, 
 and
 \item there
is $x\in (z_0,z_1)$ such that $\neg\beta_1(x)$.
\end{enumerate}
\end{description}
For each of these    cases we construct a   $\DEA$ formula
$\cond_i$ that describes it (i.e., Case $i$ holds\break iff $\cond_i$ holds)  and show that if $\cond_i$ holds, then
$\neg      [ \alpha_0, \beta_1  \dots,  \beta_{n-1}, \alpha_{n-1}, \beta_n ,\alpha_n](z_0,z_1)$
 is equivalent to  a $\DEA$ formula $\form_i$.
 Hence, $\neg      [ \alpha_0, \beta_1  \dots,  \beta_{n-1}, \alpha_{n-1}, \beta_n ,\alpha_n](z_0,z_1)$ is equivalent to $\vee_i [ \cond_i\wedge\form_i]$
 which is a $\DEA$ formula.

\noindent \textbf{Case 1}
This case is already explicitly described by  the $\DEA$ formula (in
the canonical expansion). In this case $\neg      [ \alpha_0,
\beta_1  \dots, \beta_{n-1}, \alpha_{n-1}, \beta_n
,\alpha_n](z_0,z_1)$ is equivalent to $\true$.

\noindent \textbf{Case 2}
This case is described by   a   $\DEA$ formula
$\alpha_{0}(z_0)\wedge \Alub{z}{z_0}{z_1} \beta_{1}$.
\sloppy{ In this case
 $\neg      [ \alpha_0,
\beta_1  \dots, \beta_{n-1}, \alpha_{n-1}, \beta_n
,\alpha_n](z_0,z_1)$ is equivalent to ``there is no $z\in (z_0,z_1)$
such that
  $[ \alpha_1, \beta_2  \dots,   \beta_{n} ,\alpha_{n}](z,z_1) $.''
By Corollary \ref{cor:cont}(2) this  is expressible by a $\DEA$
formula.
}

\noindent \textbf{Case 3}
The first condition of Case 3 is already explicitly described by a
$\DEA$ formula. When   the first   condition  holds,
 then the second  condition is equivalent to
 ``there is (a unique)  $r_0\in (z_0,z_1)$  such that
$r_0=\inf\{z\in (z_0,z_1)\mid \neg\beta_1(z)\}$''
(If $
\neg\KPLUS(\neg\beta_1)$ holds at $z_0$ and
 there is $x\in (z_0,z_1)$ such that $\neg\beta_1(x)$, then  such $r_0$ exists because we deal with  Dedekind complete chains.)
This $r_0$ is definable by the following  $\DEA$  formula, i.e., it
is a unique $z$ which satisfies it\footnote{We will use only existence and will not use
uniqueness.}:
\begin{equation}\label{eq:un}
\Inf (z_0, z,z_1):= z_0<z<z_1 \wedge \Alub {y}{z_0} {z}
\beta_1(y) \wedge(\neg\beta_1(z) \vee \KPLUS(\neg\beta_1)(z))
\end{equation}
Hence,   Case 3 is described by $ \alpha_0(z_0)\wedge \neg\KPLUS(\neg\beta_1)(z_0) 
\wedge \Elub{z}{z_0}{z_1} \Inf(z_0,
z,z_1)$ which is equivalent to an $\EA$ formula.


It is sufficient to show that $\Elub{z}{z_0}{z_1} \Inf(z)\wedge
\neg[ \alpha_0, \beta_1, \alpha_1,  \dots, \beta_{n+1}
,\alpha_{n+1}](z_0,z_1)$ is equivalent to a  $\DEA$ formula.

 We
prove this
 by induction on $n$.

The \emph{basis} is trivial.

\emph{Inductive step $n\mapsto n+1$.}

Define:
 \begin{align*}
 A^-_i(z_0,z):= & [\alpha_0,\beta_1,\dots ,\beta_i,\alpha_i](z_0,z) & i=1,\dots,n\\
 A^+_i(z,z_1):= & [\alpha_i,\beta_{i+1}, \dots \beta_{n+1},\alpha_{n+1}](z,z_1)& i=1,\dots,n\\
 A_i(z_0,z,z_1):= &  A^-_i(z_0,z)\wedge  A^+_i(z,z_1)& i=1,\dots,n\\
B^-_i(z_0,z):= & [\alpha_0,\beta_1,\dots
,\beta_{i-1},\alpha_{i-1},\beta_{i},\beta_i](z_0,z) & i=1,\dots,n+1\\
B^+_i(z,z_1):= &  [\beta_i , \beta_i, \alpha_{i} \beta_{i+1} \alpha_{i+1},  \dots ,\beta_{n+1},\alpha_{n+1}](z,z_1)  & i=1,\dots,n+1\\
B_i(z_0,z,z_1):= &  B^-_i(z_0,z)\wedge  B^+_i(z,z_1)  & i=1,\dots,n+1
    \end{align*}
\begin{figure} 
\begin{center}
\includegraphics[scale=0.8]{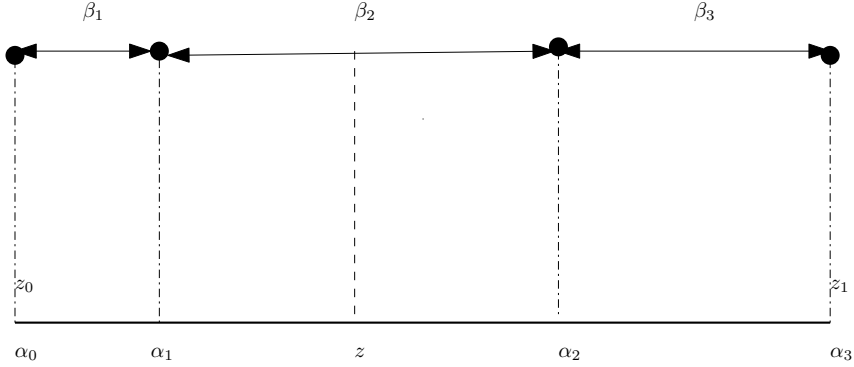}
\end{center}
 \caption{$B_2(z_0,z,z_1):=  [ \alpha_0, \beta_1 , \alpha_1,\beta_2,\beta_2](z_0,z)\wedge [\beta_2,\beta_2,\alpha_2,\beta_3,\alpha_3](z,z_1)$ }
\label{fig-b2}
\end{figure}

If the interval $(z_0,z_1)$ is non-empty, these definitions imply 
 \begin{align*}
        [ \alpha_0, \beta_1, \alpha_1,  \dots,   \beta_{n+1} ,\alpha_{n+1}](z_0,z_1)
        \Leftrightarrow
 \Alub{z}{z_0}{z_1} \big( \bigvee_{i=1}^n A_i\vee
\bigvee_{i=1}^{n+1} B_i\big)\\
 [ \alpha_0, \beta_1, \alpha_1,  \dots,   \beta_{n+1} ,\alpha_{n+1}](z_0,z_1)
        \Leftrightarrow
 \Elub{z}{z_0}{z_1} \big( \bigvee_{i=1}^n A_i\vee
\bigvee_{i=1}^{n+1} B_i\big)
 \end{align*}
Hence, for every $\varphi$
$$\Elub{z}{z_0}{z_1} \varphi(z)\wedge  \neg[ \alpha_0, \beta_1, \alpha_1,  \dots,   \beta_{n+1}
,\alpha_{n+1}](z_0,z_1)$$ is equivalent to
$$ \Elub{z}{z_0}{z_1} \big(\varphi(z) \wedge  \bigwedge_{i=1}^n \neg A_i\wedge \bigwedge_{i=1}^{n+1}
\neg B_i\big)
$$
In particular,
$$\Elub{z}{z_0}{z_1} \Inf(z)\wedge  \neg[ \alpha_0, \beta_1, \alpha_1,  \dots,   \beta_{n+1}
,\alpha_{n+1}](z_0,z_1)$$ is equivalent to
$$ \Elub{z}{z_0}{z_1} \big(  \Inf(z) \wedge \bigwedge_{i=1}^n \neg A_i\wedge \bigwedge_{i=1}^{n+1}
\neg B_i\big),
$$
where $ \Inf(z)$ was defined in equation \eqref{eq:un}.

 By the inductive assumption
 \begin{description}
 \item[(a)]
$\neg A_i$ is equivalent to a $\DEA$ formula for $i=1,\dots, n$.
\item[(b)]
$\neg B_i$ is equivalent to a $\DEA$ formula for $i=2,\dots,
 n$.
\end{description}
Recall $B_1:=B^-_1 \wedge B^+_1  $
and $B_{n+1}:=B^-_{n+1} \wedge B^+_{n+1}$.
\begin{description}
 \item[(c)]
 $\neg B^-_1$ and $ \neg B^+_{n+1}$ are
equivalent to   $\DEA$ formulas, by the induction basis.
 \item[(d)]
$\Inf(z)  \wedge \neg B^+_{1}(z,z_1) $ is equivalent to $ \Inf(z) $,
  because  if $ \Inf(z) $, then for no $x>z$, $\beta_1$
holds along  $[z,x)$.
 \item[(e)]
$\Inf(z) \wedge \neg B^-_{n+1}(z_0,z) $
is equivalent to
$\Inf(z)
\wedge (\text{``$\beta_1$ holds on $ (z_0,z)$''} \wedge \neg
B^-_{n+1}(z_0,z)) $. Since, by case 2, $ \text{``$\beta_1$ holds on $
(z_0,z)$''} \wedge \neg B^-_{n+1}(z_0,z)  $ is equivalent to a $\DEA$
formula, and $\Inf(z)$ is a $\DEA$ formula,  we conclude that
$\Inf(z)  \wedge \neg B^-_{n+1}(z_0,z) $  is
equivalent to a $\DEA$ formula.
\end{description}
Since  the set of  $\DEA$ formulas is closed under conjunction,
disjunction and $\exists$, by (a)-(e) we obtain that
$\Elub{z}{z_0}{z_1}\big(  \Inf(z)\wedge \bigwedge_{i=1}^n \neg A_i\wedge \bigwedge_{i=1}^{n+1}
\neg B_i\big)
$
is equivalent
to a  $\DEA$ formula. Therefore, $\Elub{z}{z_0}{z_1} \Inf(z)\wedge  \neg[ \alpha_0, \beta_1,
\alpha_1,  \dots, \beta_{n+1} ,\alpha_{n+1}](z_0,z_1)$
is also  a $\DEA$  formula.

This completes our proof of Lemma \ref{lem:neg1} and  of Proposition \ref{lem:neg}. 
\end{proof}

\section{Related Works} \label{sect:related}

Kamp's theorem was proved in
\begin{enumerate}
\item Kamp's thesis  \cite{Kam68} (proof $ >$ 100pages).

\item  Outlined by  Gabbay, Pnueli, Shelah   and Stavi \cite{GPSS80}
(Sect. 2)   for $\nat$  and stated that it can be extended to Dedekind
complete orders using game arguments.

\item Was proved by Gabbay \cite{Gab81} by separation arguments for $\nat$,
and extended to Dedekind complete order in \cite{GHR94}.

\item Was proved by Hodkinson \cite{Hod95} by game arguments and
simplified in \cite{Hod99} (unpublished).
\end{enumerate}
 A temporal logic has the
\emph{separation} property if its formulas can be equivalently rewritten as a boolean combination of
formulas, each of which depends only on the past, present  or future.
The separation property was introduced by  Gabbay \cite{Gab81}, and
 surprisingly,  a temporal logic which can express  $\G $ and  $\Gb $ has the separation property (over a  class $\mC$ of structures) iff it is
expressively  complete for \FOMLOiB  over $\mC$.

The separation proof for $\TL(\until,\since)$  over $\nat$  is manageable; however, over the real (and over Dedekind complete) chains it
contains many rules and transformations and   is not easy to follow.
Hodkinson and Reynolds \cite{HR06} write:
\begin{quote}
The
proofs of theorems 18 and 19 [Kamp's theorem over naturals and  over reals, respectively]  are direct,  showing that each formula can be separated. They are
tough and tougher, respectively. Nonetheless, they are effective, and so, whilst not quite
providing an algorithm  to determine  if  a set of connectives  is expressively complete, they do
suggest a potential way of telling in practice whether a given set of connectives is expressively
complete -- in Gabbay's words, \emph{try  to separate and see where
you get stuck!}
\end{quote}
The game arguments are easier to grasp, but they use complicated inductive
assertions. The proof in \cite{Hod99} proceeds roughly as follows. Let
$\mL_r$ be the set of $\TL(\until,\since)$ formulas of nesting depth
at most $r$.
 A formula of the form:  $\displaystyle \exists \bar{x} \forall y \chi(\bar{x}, y, \bar{z})$ where $ \bar{x} $ is an $n$-tuple of variables
 and $\chi$ is a quantifier free formula   over $\{<,=\}$ and  $\mL_r$-definable monadic predicates is called $\langle n,r\rangle$-decomposition formula.
 The main inductive assertion is proved by ``unusual back-and-forth games''
 and can be rephrased in logical terms  as    there is a function $f:\nat \rightarrow \nat$ such that for every
$n,r\in \nat$, the negation of   positive Boolean combinations
$\langle n,r\rangle$-decomposition formula is equivalent to a
positive Boolean combination of $\langle
f(n),(n+r)\rangle$-decomposition   formulas.

Our proof is inspired by  \cite{GPSS80} and 
\cite{Hod99}; however,
it avoids   games,
 and it     separates general logical equivalences and temporal arguments.

%
%
%
%
%

The temporal logic with the modalities $\until$ and $\since$ is not
expressively complete for
 $\IFOMLO$ over the rationals.
 Stavi introduced two additional
         modalities $ \suntil$ and $ \ssince $
        and proved that $\TLs$  is expressively
complete for $\IFOMLO$ over all linear orders \cite{GHR94}. In the forthcoming
 paper we prove Stavi's theorem.
The proof is  similar to our proof of Kamp's theorem; however, it
treats some additional cases related to gaps in orders, and replaces
$\EA$-formulas by slightly more general formulas.

\section{Future   fragment  of $\IFOMLO$ } \label{sect:future}
Many temporal formalisms studied in computer science deal  only with
future formulas,  whose truth value at any moment is determined by
what happens from a current  moment on.
 A formula (temporal, or monadic with a single free first-order  variable) $F$ is (\emph{semantically}) \emph{future}
 if for every chain $\mM$ and moment $t_0\in \mM$:
 $$\mM,t_0\models F \mbox{ iff } \mM|_{\geq t_0}, t_0\models F,$$
 where $\mM|_{\geq t_0}$ is the subchain of $\mM$ over the interval
 $[t_0,\infty)$.
For example,   $P\until Q$ and $\KPLUS(P)$ are future formulas,
while $P\since Q $ and  $\KMINUS(P)$ are not future ones.

For a set $B$ of modalities we denote by $\TL(B)$ the temporal logic
which uses only modalities from $B$. In particular, $\TL(\until)$ is
the temporal logic which uses the modality $\until$  and
$\TL(\until,\KMINUS)$ is the temporal logic  with modalities
$\until$ and $\KMINUS$.

It was shown in \cite{GPSS80}  that
Kamp's theorem holds also for \emph{future formulas} of $\IFOMLO$ over  $\omega=\langle \nat,~<\rangle$: 
\begin{thm}[Gabbay, Pnueli, Shelah, Stavi  \cite{GPSS80}] \label{thm:gpss}
 Every future $\IFOMLO$  formula  is equivalent over $\omega$-chains   to a $\TL(\until)$    formula.
\end{thm}
 The situation is radically
different for the continuous time
$\langle \real ,  
<\rangle$. In \cite{HR03} it was shown that  $\TL(\until)$ is not
expressively complete for the future fragment of $\IFOMLO$ and there
is no easy way to remedy it. In fact,  it was shown in \cite{HR03}
that   there is no temporal logic with a finite set of modalities
which is expressively equivalent to the future fragment of $\IFOMLO$
over the Reals.

From the separation proof of Kamp's theorem in
  \cite{GHR94}
  it follows
that  every future $\IFOMLO$  formula  is equivalent over Dedekind complete chains to a  $\TL(\until,\KMINUS)$ formula.

This future-past mixture of $\until$ and $\KMINUS$ is  somewhat
better than the standard $\until$ - $\since$ basis in the following
sense: although $\KMINUS$ is (like $\since$) a past modality, it
does not depend on much of the past.  The formula $\KMINUS(P)$
depends just on an arbitrarily short `near past', and is actually
{independent of most of the past}. In this sense, we may say that it
is  an ``almost''   future formula.

\begin{defi}[Syntactically future $\TL(\until,\KMINUS)$ formulas]
A $\TL(\until,\KMINUS)$ formula is syntactically future if it is a
boolean combination of atomic formulas and formulas of the form
$\vp_1\until \vp_2$, where $\vp_1$ and $\vp_2$ are arbitrary
$\TL(\until,\KMINUS)$ formulas.
\end{defi}
The following  lemma immediately follows from the definition and the
observation that  $\mM|_{\geq t_0}, t_0\models \neg \KMINUS (\vp)$.
\begin{lem}\label{lem:synt-fut-uk0}
A syntactically future  $\TL(\until,\KMINUS)$ formula is future. A
$\TL(\until,\KMINUS)$ formula is future iff it is equivalent to a
syntactically future  $\TL(\until,\KMINUS)$ formula.
\end{lem}
The next  theorem (implicitly) appears in \cite{GHR94} (Chapter 8).
\begin{thm}\label{thm:ghr94}
 Every future $\IFOMLO$  formula  is equivalent over Dedekind complete chains  to a syntactically future $\TL(\until,\KMINUS)$    formula.
\end{thm}
Since $\KMINUS \vp$ is equivalent to $\false$ over discrete orders, we
obtain that Theorems \ref{thm:gpss} is an instance of Theorem
\ref{thm:ghr94}.

Theorem \ref{thm:ghr94}  is easily obtained by a slight refinement
of  our proof of Kamp's theorem. We outline its proof  in the rest
of this section.

\begin{defi}[$(z_0,z_1)$-$\EA$ formula] Let $z_0$ and $z_1$ be two variables. A formula  $z_0>z_1$, $z_0=z_1$  or of the
form  $[ \alpha_0, \beta_1  \dots,  \beta_{n-1}, \alpha_{n-1},
\beta_n ,\alpha_n](z_0,z_1)$ is called a $(z_0,z_1)$-$\EA$ formula.
 A formula is $(z_0,z_1)$-$\DEA$ formula if it is
equivalent to a disjunction of $(z_0,z_1)$-$\EA$ formulas.
\end{defi}

 \begin{lem}[closure properties]\label{lem:fut-closure}
The set of $(z_0,z_1)$-$\DEA$ formulas is  closed under disjunction
and  conjunction. If $\vp_1$ is a  $(z_0,z_1)$-$\DEA$ formula and
$\vp_2$ is a  $(z_1,z_2)$-$\DEA$ formula, then $\Elub{z_1}{z_0}{z_2}
(\vp_1\wedge \vp_2 )$ is a $(z_0,z_2)$-$\DEA$ formula.
\end{lem}

The set of $(z_0,z_1)$-$\DEA$  formulas is not closed under
negation. However, we show     that the negation of a  $\DEA$
formula is equivalent to a $(z_0,z_1)$-$\DEA$  formula in the
expansion of the chains by all
 $\TL(\until,\KMINUS)$ definable predicates.

\begin{defi}[The canonical $\TL(\until,\KMINUS)$ and $\TL(\since,\KPLUS)$  expansions]  Let $\mM$ be a $\Sigma$ chain. We denote
 by  $\mE[\Sigma,\TL(\until,\KMINUS)]$  the set of unary predicate names
$\Sigma\cup\{A~\mid A$ is an  $\TL(\until,\KMINUS)$ formula over
$\Sigma ~\}$.
 The canonical $\TL(\until,\KMINUS)$-expansion of $\mM$ is  an expansion of
$\mM$ to an $\mE[\Sigma,\TL(\until,\KMINUS)]$-chain,  where each
predicate's name $A\in  \mE[\Sigma,\TL(\until,\KMINUS)]$ is
interpreted as $\{a\in \mM\mid \mM,a\models A\}$.
 We say that first-order formulas
in the signature $\mE[\Sigma,  \TL(\until,\KMINUS)] \cup\{<\}$ are
equivalent over $\mM$ (respectively, over a class of $\Sigma$-chains
$\mC$)  if they are equivalent in the canonical
$\TL(\until,\KMINUS)$-expansion of $\mM$ (in  the canonical
$\TL(\until,\KMINUS)$-expansion of every $\mM\in \mC$).
The canonical $\TL(\since,\KPLUS)$-expansion of a chain $\mM$ is
defined similarly.
\end{defi}

The next lemma  implies  that Lemma \ref{lem:neg1} holds over  the
canonical $\TL(\since,\KPLUS)$-ex\-pan\-sions and over canonical
 $\TL(\until,\KMINUS)$-expansions.
\begin{lem}\label{lem:f-neg1}\hfill
\begin{enumerate}
\item
 $\neg      [ \alpha_0, \beta_1  \dots,  \beta_{n-1}, \alpha_{n-1}, \beta_n ,\alpha_n](z_0,z_1)$
is equivalent   over the canonical $\TL(\since,\KPLUS)$-expansions
of Dedekind complete chains to a $(z_0,z_1)$-$\DEA$-formula.
\item Dually,
 $\neg      [ \alpha_0, \beta_1  \dots,  \beta_{n-1}, \alpha_{n-1}, \beta_n ,\alpha_n](z_0,z_1)$
is equivalent   over the canonical\break $\TL(\until,\KMINUS)$-expansions
of Dedekind complete chains to a $(z_0,z_1)$-$\DEA$-formula.
\end{enumerate}
%
\end{lem}
\begin{proof}
Actually, our proof  of Lemma  \ref{lem:neg1}, as it is, works for
 the
canonical  $\TL(\since,\KPLUS)$-expansions of  Dedekind complete
chains, when ``$\EA$ formulas" are replaced by ``$(z_0,z_1)$-$\EA$ formulas.''

Indeed, Lemma \ref{lem:occ} uses only modality $\KPLUS$. Thus,
exactly the same proof works for $\TL(\since,\KPLUS)$-expansions and
for  $\TL(\until,\KMINUS)$-expansions (because $\KPLUS$ is
equivalent to a $\TL(\until)$ formula).

In the proof of Corollary \ref{cor:cont}(1) we used Lemma
\ref{lem:occ}  and $\until$  modality. Hence, it holds for
 $\TL(\until,\KMINUS)$-expansions.
 Corollary \ref{cor:cont}(2) is dual and it holds for
 $\TL(\since,\KPLUS)$- expansions.

In proof of Lemma \ref{lem:neg1} we use standard logical
equivalences and Corollary \ref{cor:cont}(2). Hence, it works as it
is for the canonical $\TL(\since,\KPLUS)$- expansions of Dedekind
complete chains. This proves (1). Item (2) is the mirror image of
(1).
\end{proof}

\begin {nota} \label{denote:f-simple}
    We use the abbreviated notation
$ [ \alpha_0, \beta_1  \dots,   \alpha_{n-1},
 \beta_n,\alpha_n, \beta_{n+1}](z_0,\infty)$ for
 $$\eqalign{\E x_n \dots \E x_1 \E x_0
  z_0=x_0                          &\land~ \left(x_n >   \dots
> x_1 > x_0\right)\cr
                         &\land ~\AND _{j=0}^{n} \alpha_j(x_j)
                         \land~
                             \AND _{j=1}^{n}  \intprop {\beta_j}
                             {x_{j-1}}  {x_{j}}
                           \wedge ~(\forall y)_{>x_n}
                        {\beta_{n+1}} (y);}
$$ such formulas will be called $(z_0,\infty)$-formulas; we
use the similarly abbreviated notation
    $ [\beta_0, \alpha_0, \beta_1  \dots,   \alpha_{n-1}, \beta_n,\alpha_n](-\infty,z_0)$
    for the $\EA$-formula
$$ \eqalign{\E x_n \dots \E x_1 \E x_0
  z_0=x_n                          &\land~ \left(x_n > \dots
> x_1 > x_0\right)\cr
                        & \wedge ~\AND _{j=0}^{n} \alpha_j(x_j)
                           \wedge~
                             \AND _{j=1}^{n}  \intprop {\beta_j}
                             {x_{j-1}}  {x_{j}}
                          \wedge~\Intprop {\beta_0} {x_{0}}.   }                                                                                                          %
$$
 \end {nota}

\begin{lem}\label{lem:fut-inf}
 $  [ \alpha_0, \beta_1  \dots,   \alpha_{n-1},
 \beta_n,\alpha_n, \beta_{n+1}](z_0,\infty)$
%
 over canonical  $\TL(\until,\KMINUS)$-ex\-pan\-sions
is equivalent
  to a  $\TL(\until,\KMINUS)$ formula.

%
%
\end{lem}
\begin{proof}
By a straightforward formalization as in the proof of Proposition
\ref{prop:form}.
%
%
%
\end{proof}

\begin{defi}[Syntactically future $\IFOMLO$  formulas]
A $\IFOMLO$  formula $\vp(z_0) $  is syntactically future if all its
quantifiers are bounded quantifiers of the form $(\forall y)_{>z_0}
$ and $(\exists y)_{>z_0} $.
\end{defi}
The following  lemma immediately follows from the definition.
\begin{lem} \label{lem:synt-fut-fo}
A syntactically future  $\IFOMLO$  formula is future. A $\IFOMLO$
formula $\vp(z_0)$  is future iff it is equivalent to a
syntactically future $\IFOMLO$  formula.
\end{lem}
\begin{defi} Let $(z_0,z_1,\dots ,z_k)$ be a
sequence of distinct variables.
 A formula
is $(z_0,z_1,\dots ,z_k, \infty)$-$\EA$ formula if it is a
conjunction $\bigwedge_{i\leq k} \vp_i$, where $\vp_k$ is
$(z_k,\infty)$-$\EA$ formula and $\vp_i$ is $(z_i,z_{i+1})$-$\EA$
formulas for $i<k$. A formula is a $(z_0,z_1,\dots ,z_k,
\infty)$-$\DEA$ formula if it is equivalent to a disjunction of
$(z_0,z_1,\dots ,z_k, \infty)$-$\EA$ formulas.
\end{defi}

\begin{lem}\label{lem:from:fo-to-uk} Let $\vp(z_0,z_1,\dots ,z_k)$ be a  $\IFOMLO$  formula
with  free variables in $\{z_i\mid i\leq k\}$  and all its
quantifiers are bounded quantifiers of the form $(\forall y)_{>z_0}
$ and $(\exists y)_{>z_0} $. Then, there is $(z_0,z_1,\dots ,z_k,
\infty)$-$\DEA$ formula $\psi$ such that $z_0<z_1<\dots <z_k\wedge
\vp$ is equivalent  over the canonical
$\TL(\until,\KMINUS)$-expansions of Dedekind complete chains to
$z_0<z_1<\dots <z_k\wedge \psi$.
\end{lem}
\begin{proof} By  Lemmas \ref{lem:fut-closure},  \ref{lem:f-neg1}(2), \ref{lem:fut-inf} and a
straightforward structural induction.
\end{proof}
Now we are ready to prove   Theorem \ref{thm:ghr94}.
\begin{proof} (of Theorem \ref{thm:ghr94}) Assume that  $\vp(z_0)$ is a future $\IFOMLO$ formula.
By Lemma \ref{lem:synt-fut-fo} w.l.o.g we can assume that all its
quantifiers are  bounded quantifiers of the form $(\forall y)_{>z_0}
$ and $(\exists y)_{>z_0} $. By \ref{lem:from:fo-to-uk},  it is
equivalent to $(z_0,\infty)$-$\DEA$ formula. Hence, by Lemma
\ref{lem:fut-inf}, it is equivalent to a $\TL(\until,\KMINUS)$
formula. Therefore, by Lemma  \ref{lem:synt-fut-uk0} it is
equivalent to a syntactically future $\TL(\until,\KMINUS)$ formula.
\end{proof}

In \cite{PR12} we erroneously stated that the analog of Proposition
\ref{prop:fo2ea} holds  for $\TL(\until,\KMINUS)$-expansions.
However,  ``$P$ is unbounded from below'' is expressible by  a
$\IFOMLO$ sentence $\forall x \exists y (y<x \wedge P(y))$;  yet
  there is no $\DEA$ formula which expresses ``$P$ is
unbounded from below" over the canonical
$\TL(\until,\KMINUS)$-expansions of integer-chains. 
We  state the following Proposition
  for the sake of completeness.

\begin{prop}\label{lem:sake-comp}
Every $\IFOMLO$ formula is equivalent   over the canonical
$\TL(\until,\KMINUS)$-expansions of Dedekind complete chains to a
positive boolean combination  of $\EA$ formulas and sentences  of
the form ``$P$ is unbounded from below.''
\end{prop}
The additional  step needed for the proof of  Proposition
\ref{lem:sake-comp} is an observation that
 $\neg  [\beta_0, \alpha_0, \beta_1  \dots,
\alpha_{n-1},
 \beta_n,\alpha_n](-\infty,z_0)$ is equivalent
over the canonical $\TL(\until,\KMINUS)$-ex\-pan\-sions of Dedekind
complete chains to a positive boolean combinations of $(-\infty,z_0)$ -$\EA$ formulas
and sentences of the form ``$P$ is unbounded from below.''
   This is proved  almost in the same way as  Lemma
\ref{lem:neg1}.

\bigskip \noindent \textbf{{Acknowledgments}} I am  very  grateful  to
Yoram Hirshfeld for numerous insightful  discussions,
and to the
anonymous referees for their helpful suggestions.

\end{document}